\documentclass[11pt]{article}
\usepackage{amsmath,amsfonts, amssymb, amsthm}
\usepackage{graphicx}
\usepackage[margin=3cm]{geometry}
\def\idty{{\mathchoice {\mathrm{1\mskip-4mu l}} {\mathrm{1\mskip-4mu l}} %
{\mathrm{1\mskip-4.5mu l}} {\mathrm{1\mskip-5mu l}}}}

\newtheorem{theorem}{Theorem}[section]

\newtheorem{lemma}[theorem]{Lemma}
\newtheorem{remark}[theorem]{Remark}

\newtheorem{corollary}[theorem]{Corollary}

\newcommand{\bR}{{\mathbb R}}
\newcommand{\bC}{{\mathbb C}}

\newcommand{\bZ}{{\mathbb Z}}
\newcommand{\cA}{{\cal A}}
\newcommand{\cB}{{\cal B}}
\newcommand{\cH}{{\cal H}}
\newcommand{\rd}{{\rm d}}
\newcommand{\supp}{\operatorname{supp}}
\newcommand{\cU}{{\cal U}}
\newcommand{\cL}{{\cal L}}

\begin{document}

\title{Lieb-Robinson Bounds \\[5pt]
for Harmonic and Anharmonic Lattice Systems}

\author{Bruno Nachtergaele${}^1$, Hillel Raz${}^1$, Benjamin Schlein${}^2$, \, and Robert Sims${}^3$\\[5pt]
bxn@math.ucdavis.edu, hraz@math.ucdavis.edu\\ schlein@math.lmu.de, rsimsdrb@gmail.com
\\[15pt]
1. Department of Mathamatics, University of California at Davis, \\
Davis, CA 95616, USA\\
\\
2. Institute of Mathematics, University of Munich, \\
Theresienstr. 39, D-80333 Munich, Germany \\ \\
3. Faculty of Mathematics, University of Vienna \\
Nordbergstr. 15, A-1090 Vienna, Austria}

\maketitle

\begin{abstract}
We prove Lieb-Robinson bounds for systems defined on infinite dimensional
Hilbert spaces and described by unbounded Hamiltonians. In particular,
we consider harmonic and certain anharmonic lattice systems.
\end{abstract}

\renewcommand{\thefootnote}{$ $}
\footnotetext{Copyright \copyright\ 2008 by the authors. This article may be
reproduced in its entirety for non-commercial purposes.}
\section{Introduction}

An important class of systems in statistical mechanics is described
by the (an)harmonic lattice Hamiltonians, which have a continuous degree
of freedom, thought of as a particle trapped in a potential, at each site
of a lattice. The particles interact by a linear or non-linear force.
For example, such models are thought to describe the emergence of macroscopic
non-equilibrium phenomena, such as heat conduction, from many-body Hamiltonian
dynamics \cite{spohn2006,aoki2006}, the understanding of which is one of the
long-standing open problems in mathematical statistical mechanics
\cite{bonetto2000}.

In terms of technical difficulty, lattice oscillator models are intermediate
between spin systems, where the degrees of freedom, each described by
a finite-dimensional Hilbert space, are labeled by a discrete set, usually
a lattice such as $\mathbb{Z}^\nu$, on the one hand, and particles in
continuous space, which necessarily have an infinite-dimensional state space,
on the other hand. Even in the classical case lattice oscillator systems
are significantly more difficult to study than spin systems, and also for
them more is known than for particle models in the continuum. E.g., the
existence of the dynamics in the thermodynamics limit was studied by Lanford,
Lebowitz, and Lieb in \cite{LLL77}.

In this paper we focus on an essential locality property of the dynamics of
quantum harmonic and anharmonic lattice models. Since these are non-relativistic
models there is no {\em a priori} bound on the speed of propagation of
signals in these systems. In the case of quantum spin systems with finite-range
interactions, Lieb and Robinson \cite{LR1} showed that there is nevertheless an upper
bound on the speed of propagation in the sense that disturbances in the system
remain confined in a ``light'' cone up to small corrections that decay
at least exponentially fast away from the light cone. This is the so-called
Lieb-Robinson bound which is an upper bound on the speed of propagation.

In the past few years several generalizations, improvements, and applications
of Lieb-Robinson type bounds have appeared. This work can be regarded
as one further extension, going for the first time beyond the realm of
quantum spin systems. Here, by quantum spin system we mean any quantum system
with a finite dimensional Hilbert space of states. For example, a quantum
spin system over a finite subset $\Lambda \subset \bZ^{\nu}$ is described on
the Hilbert space
\[
\cH_{\Lambda} = \bigotimes_{x\in \Lambda} \cH_x \qquad \text{with }
\cH_x = \bC^{n_x}
\]
where the dimensions $2 \leq n_x < \infty$ are related to the magnitude of the
spin at site $x \in \Lambda$. The algebra of observables for this quantum spin
system is then given by
\[
\cA_{\Lambda} = \bigotimes_{x \in \Lambda} \cB (\cH_x) = \cB ( \cH_{\Lambda})
\]
where $\cB (\cH_x)$ is the space of bounded operators on $\cH_x$ (that is the
space of all $n_x \times n_x$ matrices). The Hamiltonian of the quantum spin
system is usually written in the form
\[
H_{\Lambda} = \sum_{X \subset \Lambda} \Phi (X)
\]
where the interaction $\Phi: 2^{\Lambda} \to \cA_{\Lambda}$ is such that
$\Phi(X)^* = \Phi (X) \in \cA_X = \otimes_{x \in X} \cB (\cH_x)$ for all
$X \subset \Lambda$. The time evolution associated with the Hamiltonian
$H_{\Lambda}$ is then the one-parameter group of automorphisms
$\{\tau^{\Lambda}_t \}_{t \in \bR}$ defined by
\[
\tau_t^{\Lambda} (A) = e^{itH_{\Lambda}} A e^{-itH_{\Lambda}}
\qquad \text{for all } A \in \cA_{\Lambda} \,.
\]

For such systems, under appropriate conditions on the interactions $\Phi (X)$
 (short-range conditions) it was first proved by Lieb and Robinson in
\cite{LR1}, that, given $A \in \cA_X$, $B \in \cA_Y$,
\begin{equation}  \label{eq:LR1}
\| [ \tau^{\Lambda}_t (A), B ] \|  \leq C \| A \| \| B \| \;
e^{- \mu(d(X,Y) - v |t|)}
\end{equation}
where $d(X,Y) = \min_{x \in X, y \in Y} |x-y|$ and $|x| =\sum_{j=1}^{\nu} |x_j|$.
The physical interpretation of this inequality is
straightforward; if two observables $A$ and $B$ are supported in disjoint
regions, then even after evolving the observable $A$, apart from exponentially
small contributions, their supports remain essentially  disjoint up to times
$t \leq d(X,Y)/ v$. In other words, this bound asserts that the speed of
propagation of perturbations in quantum spin systems is bounded.

In the original proof of the Lieb-Robinson bounds (see \cite{LR1}), the constant
$C$ and the velocity $v$ on the right hand side of (\ref{eq:LR1}) depended
in a crucial way on $N = \text{max}_{x \in \Lambda} n_x $, the maximal dimension
of the different spin spaces. More recently, new Lieb-Robinson bounds of the form
(\ref{eq:LR1}) were derived with a constant $C$ and a velocity of propagation
$v$ independent of the dimension of the various spin spaces \cite{HK,NOS1}.
This new version of the Lieb-Robinson bounds allowed for new applications,
such as, among other results, a proof of the Lieb-Schutz-Mattis theorem in
higher dimension, see \cite{H,NS2}.

It seems natural to ask whether Lieb-Robinson bounds such as (\ref{eq:LR1})
can be extended to systems defined on infinite dimensional Hilbert spaces,
and described by unbounded Hamiltonians. Although the constant $C$ and the
velocity $v$ in (\ref{eq:LR1}) are independent of the dimension
of the spin spaces, they depend on the operator norm of the interactions
$\Phi (X)$; for this reason, if one deals with unbounded Hamiltonians, the
methods developed in \cite{NOS1,NS1, HK} cannot be applied directly.
Nevertheless, in the present paper we prove that Lieb-Robinson bounds
can be established for three different types of models with unbounded
Hamiltonians, which we now present. For the precise statements see
Sections \ref{sec:phi}, \ref{sec:harm}, and \ref{sec:anharm}.

First, in Section \ref{sec:phi}, we consider systems defined on
an infinite dimensional Hilbert space by Hamilton operators with possibly
unbounded on-site terms but bounded interactions between sites. In this
case, we show that the analysis of \cite{NOS1} goes through with only minor
changes, and that Lieb-Robinson bounds can be proven in quite a large
generality (see Theorem \ref{thm:phi}). A class of interesting examples
of this are lattice oscillators coupled by bounded interactions.
For a finite subset $\Lambda \subset \bZ^{\nu}$, one considers the system defined on
the
Hilbert space $\cH_{\Lambda} = \bigotimes_{x \in \Lambda} L^2 ( \bR, \rd q_x)$ by the Hamiltonian
\[
H = \sum_{x \in \Lambda} p_x^2 + V (q_x) + \sum_{x,y\in\Lambda,\, |x-y| = 1} \phi (q_x -q_y)
\]
where $p_x = -i \, \rd / \rd q_x$, the real function $V$ is such that $-\Delta_q + V (q)$ is a
self-adjoint operator, and $\phi \in L^{\infty} (\bR)$ is real valued. Another
commonly studied model that satisfies the conditions of this result is the so-called
quantum rotor Hamiltonian of the form
\begin{equation*}
H = -\sum_x \frac{\partial^2}{\partial\theta_x^2} +\sum_{x,y} J_{xy}
\cos(\theta_x - \theta_y +\phi)\,
\end{equation*}
where $\theta_x$ is the angle associated with the rotor at site $x$, and $J_{xy}$
are coupling constants assumed to vanish whenever $|x-y|$ exceeds a finite
range $R$. Quantum rotor Hamiltonians are used to study a variety of physical situations
such as Josephson junction arrays \cite{alsaidi2003}, the Bose-Hubbard model \cite{polak2007},
and crystals consisting of molecules with rotor degrees of freedom \cite{gregor2006}.

Second, in Section~\ref{sec:harm}, we consider harmonic lattice systems
for which the Hamiltonian describes a system of linearly coupled harmonic
oscillators situated at the points of a finite subset $\Lambda \subset \bZ^{\nu}$.
The standard Hamiltonian is of the form
\begin{equation*}
H^h \, =\,  \sum_{ x} p_{ x }^2 \, +\, \omega^2 \, q_{ x}^2 \, + \,
\sum_{|x-y|=1}\sum_{j = 1}^{\nu}  \lambda_j \, (q_{ x } - q_{ y})^2,
\end{equation*}
defined on a finite hypercube in $\mathbb{Z}^\nu$, with periodic boundary
conditions.
In this case, not only the on-site terms but also the interactions between
sites are given by unbounded operators, and the analysis of \cite{NOS1} cannot
be applied. As is well-known, the time evolution for harmonic systems can be
computed explicitly (see Lemma~\ref{lem:weylevo}), and the derivation of
Lieb-Robinson bounds (in the form given in Theorem~\ref{thm:harlrb}) reduces
to the study of the asymptotic properties of certain Fourier sums
(see Lemma~\ref{lm:htx}).

Finally, in Section~\ref{sec:anharm}, we consider local anharmonic perturbations
of the harmonic lattice system of the form
\begin{equation*}
H = \sum_{x} p_x^2 \, + \, \omega^2 \, q_x^2 \, +
\, \sum_{|x-y|=1} \sum_{j = 1}^{\nu}  \lambda_j \, (q_{ x } - q_{ y})^2 \, +\,
\sum_{x} V (q_x)\,.
\end{equation*}
Assuming that the local perturbation $V$ is sufficiently weak (in an appropriate sense),
and making use of an interpolation argument between the harmonic and the anharmonic time-
evolution, we derive Lieb-Robinson bounds in Theorem~\ref{thm:anharm}.

Next, we discuss the classes of observables for which we obtain the Lieb-Robinson
bounds in each of the three types of models. In the case of quantum spin
systems, i.e., the case where the Hilbert spaces associated with a lattice site
are all finite-dimensional, one proves Lieb-Robinson bounds for a pair of
arbitrary observables $A$ and $B$ with finite supports (see (\ref{eq:LR1})).
It is not clear in general that such a result should be expected when the Hilbert
spaces are infinite-dimensional and the Hamiltonians unbounded. If the
unboundedness in the Hamiltonian is restricted to on-site terms while
interactions between sites are bounded and of sufficiently short range,
the standard Lieb-Robinson bound can be derived for arbitrary bounded
observables. This is explained in Section \ref{sec:phi}. The novelty of
this paper concerns harmonic and anharmonic lattice systems which have
unbounded interactions of the form $(q_x-q_y)^2$. In Section~\ref{sec:harm}
and Section~\ref{sec:anharm} we prove Lieb-Robinson bounds for Weyl operators.
The main advantage of working in the Weyl algebra is a consequence of the fact
that the class of Weyl operators is invariant under the dynamics of the
harmonic lattice model, a property that is also used in our treatment of
anharmonic models. The Lieb-Robinson bounds that we obtain for the Weyl
operators are sufficient to derive bounds for more general observables, such
as $q_x$ and $p_x$ as well as compactly supported smooth bounded functions
of $q_x$ and $p_x$. This is discussed in Section \ref{sec:disc}.

Note that locality bounds for harmonic and anharmonic lattice systems have already
been obtained in the classical setting; while harmonic systems are well-understood,
anharmonic lattice systems are much more complicated, and a full understanding,
even in the classical case, has not been reached, yet. In \cite{MPPT}, Marchioro,
Pellegrinotti, Pulvirenti, and Triolo considered anharmonic systems in thermal
equilibrium and proved that, after time $t$, the influence of local perturbations
becomes negligible at distances larger than $t^{4/3}$. These bounds were recently
improved in \cite{BCDM} by Butt\`a, Caglioti, Di Ruzza, and Marchioro, who proved
that after time $t$ local perturbations of thermal equilibrium are exponentially
small in $\log^2 t$ at distances larger than $t \log^{\alpha} t$.

In the quantum mechanical setting, on the other hand, we are only aware of the
recent work of Buerschaper, who derived, in \cite{B}, Lieb-Robinson type bounds for
harmonic lattice systems.

\section{Lieb-Robinson Estimates for Hamiltonians with Bounded Non-Local Terms}
\label{sec:phi}

\setcounter{equation}{0}

In this section, we will state and prove our first example of Lieb-Robinson
estimates for systems with unbounded Hamiltonians. We consider here the dynamics
generated by unbounded Hamiltonians, assuming, however, the unbounded interactions
to be completely local. It turns out that, for such systems, locality bounds can be
proven in the same generality as for quantum spin systems (see Theorem
\ref{thm:phi} below). Moreover, the proof of this result only requires minor
modifications with respect to the arguments presented in \cite{NOS1}.

We first introduce the underlying structure on which our models will be defined.
Let $\Gamma$ be an arbitrary set of sites equipped with a metric $d$.
For $\Gamma$ with infinite cardinality, we will need to assume that there exists a
non-increasing function $F: [0, \infty) \to (0, \infty)$ for which:

\noindent i) $F$ is uniformly integrable over $\Gamma$, i.e.,
\begin{equation} \label{eq:fint}
\| \, F \, \| \, := \, \sup_{x \in \Gamma} \sum_{y \in \Gamma}
F(d(x,y)) \, < \, \infty,
\end{equation}

\noindent and

\vspace{.3cm}

\noindent ii) $F$ satisfies
\begin{equation} \label{eq:intlat}
C \, := \, \sup_{x,y \in \Gamma} \sum_{z \in \Gamma}
\frac{F \left( d(x,z) \right) \, F \left( d(z,y)
\right)}{F \left( d(x,y) \right)} \, < \, \infty.
\end{equation}

Given such a set $\Gamma$ and a function $F$,  it is easy to see that
for any $a \geq 0$ the function
\begin{equation*}
F_a(x) = e^{-ax} \, F(x),
\end{equation*}
also satisfies i) and ii) above with $\| F_a \| \leq \| F \|$ and $C_a \leq C$.

In typical examples, one has that $\Gamma  \subset \mathbb{Z}^{\nu}$ for
some integer $\nu \geq 1$, and the metric is
just given by $d(x,y) = |x -  y|=\sum_{j=1}^{\nu} |x_j - y_j|$.
In this case, the function $F$ can be
chosen as $F(|x|) = (1 + |x|)^{- \nu - \epsilon}$ for any $\epsilon >0$.

To each $x \in \Gamma$, we will associate a Hilbert space $\mathcal{H}_x$.
Unlike in the setting of quantum spin systems,
we will not assume that these Hilbert spaces are finite dimensional.
For example, in many relevant systems, one considers
$\mathcal{H}_x = L^2( \mathbb{R}, \rd q_x)$.
With any finite subset $\Lambda \subset \Gamma$,
the Hilbert space of states over $\Lambda$ is given by
\begin{equation*}
\mathcal{H}_{\Lambda} \, = \, \bigotimes_{x \in \Lambda} \mathcal{H}_x,
\end{equation*}
and the local algebra of observables over $\Lambda$ is then defined to be
\[
\mathcal{A}_{\Lambda} = \bigotimes_{x \in \Lambda} \cB (\cH_x),
\]
where $\cB (\cH_x)$ denotes the algebra of bounded linear operators on $\cH_x$.

If $\Lambda_1 \subset \Lambda_2$, then there is a natural way of identifying
$\mathcal{A}_{\Lambda_1} \subset \mathcal{A}_{\Lambda_2}$, and (also in the
case of infinite $\Gamma$) we may therefore define the algebra of
local observables by the inductive limit
\begin{equation*}
\mathcal{A}_{\Gamma} \, = \, \bigcup_{\Lambda \subset \Gamma} \mathcal{A}_{\Lambda},
\end{equation*}
where the union is over all finite subsets $\Lambda \subset \Gamma$; see
\cite{bratteli1987,bratteli1997} for a general discussion of these topics.

For the locality results we wish to describe, the notion of support of
an observable will be important. The support of an observable
$A \in \mathcal{A}_{\Lambda}$ is the
minimal set $X \subset \Lambda$ for which $A = A' \otimes \idty$ for some
$A' \in \mathcal{A}_X = \bigotimes_{x \in X} \cB (\cH_x)$.

The result discussed in this section corresponds to bounded perturbations of
local self-adjoint Hamiltonians. We fix a collection of local operators
$H^{\rm loc} = \{ H_x \}_{x \in \Gamma}$ where each $H_x$ is a self-adjoint
operator over $\mathcal{H}_x$. Again, we stress that these
operators $H_x$ need {\it not} be bounded.

In addition, we will consider a general class of bounded perturbations.
These are defined in terms of an interaction $\Phi$, which is a map from the
set of subsets of $\Gamma$ to $\mathcal{A}_{\Gamma}$ with the property that
for each finite set $X \subset \Gamma$, $\Phi(X) \in \mathcal{A}_X$ and
$\Phi(X) ^*= \Phi(X)$. To obtain our bound, we need to impose a growth
restriction on the set of interactions $\Phi$ we consider.
For any $a \geq 0$, denote by $\mathcal{B}_a(\Gamma)$ the set of interactions
for which
\begin{equation} \label{eq:defphia}
\| \Phi \|_a \, := \, \sup_{x,y \in \Gamma}  \frac{1}{F_a (d(x,y))} \,
\sum_{X \ni x,y} \| \Phi(X) \| \, < \, \infty.
\end{equation}

Now, for a fixed sequence of local Hamiltonians $H^{\rm loc} = \{H_x \}$, as
described above, an interaction $\Phi \in \mathcal{B}_a(\Gamma)$, and a finite subset
$\Lambda \subset \Gamma$, we will consider self-adjoint Hamiltonians of the form
\begin{equation} \label{eq:localham}
H_{\Lambda} \, = \, H^{\rm loc}_{\Lambda} \, + \, H^{\Phi}_{\Lambda} \, = \, \sum_{x \in \Lambda} H_x \, + \, \sum_{X \subset \Lambda} \Phi(X),
\end{equation}
acting on $\mathcal{H}_{\Lambda}$ (with domain given by $\bigotimes_{x \in \Lambda} D(H_x)$ where $D(H_x) \subset \cH_x$ denotes the domain of $H_x$). As these operators are self-adjoint, they generate a dynamics, or time evolution, $\{ \tau_t^{\Lambda} \}$,
which is the one parameter group of automorphisms defined by
\begin{equation*}
\tau_t^{\Lambda}(A) \, = \, e^{it H_{\Lambda}} \, A \, e^{-itH_{\Lambda}} \quad \mbox{for any} \quad A \in \mathcal{A}_{\Lambda}.
\end{equation*}
For Hamiltonians of the form (\ref{eq:localham}), we have a bound analogous to (\ref{eq:LR1}), see Theorem~\ref{thm:phi} below.

Before we present this result, we make an observation. It seems intuitively
clear that the spread of interactions through a system should
depend on the surface area of the support of the local observables being evolved;
not their volume. One can make this explicit by introducing the following
notation. Denote the surface of a set $X$, regarded as a
subset of $\Lambda \subset \Gamma$, by
\begin{equation} \label{eq:defsurf}
S_{\Lambda}(X) \, = \, \left\{ Z \subset \Lambda \, : \, Z \cap X \neq \emptyset
\mbox{ and }  Z \cap X^c \neq \emptyset \right\}.
\end{equation}
Here we will use the notation $S(X) =S_{\Gamma}(X)$, and define
the $\Phi$-boundary of a set $X$, written $\partial_{\Phi} X$, by
\begin{equation*}
\partial_{\Phi} X \, = \, \left\{ x \in X \, : \, \exists Z \in S(X)
  \mbox{ with } x \in Z \mbox{ and } \Phi(Z) \neq 0 \, \right\}.
\end{equation*}

We have the following result.

\medskip

\begin{theorem}\label{thm:phi} Fix a local Hamiltonian $H^{\rm loc}$ and an interaction $\Phi \in \mathcal{B}_a(\Gamma)$ for some
$a \geq 0$. Let $X$ and $Y$ be subsets of $\Gamma$. Then, for
any $\Lambda \supset X \cup Y$ and any pair of local observables $A \in \mathcal{A}_X$ and $B \in \mathcal{A}_Y$, one has that
\begin{equation} \label{eq:lrbd1}
\left\| [ \tau_t^{\Lambda}(A), B ] \right\| \, \leq \, \frac{2 \, \| A \|
\, \|B \|}{C_a} \, g_a(t) \, D_a(X,Y),
\end{equation}
where
\begin{equation} \label{eq:gat}
g_a(t) = \left\{ \begin{array}{cc} e^{2 \| \Phi \|_a C_a |t|} - 1 &
\mbox{if } d(X,Y) >0, \\ e^{2 \| \Phi \|_a C_a |t|} &
\mbox{otherwise,} \end{array} \right.
\end{equation}
and $D_a(X,Y)$ is given by
\begin{equation} \label{eq:defda}
D_a(X,Y) =
\min \left[ \sum_{x \in \partial_\Phi X} \sum_{y \in
  Y} \, F_a \left( d(x,y) \right), \sum_{x \in X} \sum_{y \in
  \partial_\Phi Y} \, F_a \left( d(x,y) \right)\right].
\end{equation}
\end{theorem}

The following corollary provides a bound in terms of $d(X,Y) = \min_{x\in X, y \in Y} d(x,y)$,
the distance between the supports $X,Y$.

\begin{corollary}\label{cor:phi}
Under the same assumptions as in Theorem \ref{thm:phi}, we have
\begin{equation} \label{eq:vel}
\left\| [ \tau_t^{\Lambda}(A), B ] \right\| \, \leq \, \frac{2 \, \| A \|
\, \|B \| \, \|F\|}{ C_a} \, \min \left[ \left|
    \partial_{\Phi}X \right|, \left| \partial_{\Phi}Y \right| \right] \, e^{- a \,\left[
 d(X,Y) - \frac{2 \| \Phi \|_a C_a}{a} |t| \right]},
\end{equation}
\end{corollary}

\begin{proof}[Proof of Theorem \ref{thm:phi}]
For any finite $Z \subset \Gamma$, we introduce the quantity
\begin{equation} \label{def:cbxt}
C_B(Z; t) \, := \, \sup_{A \in \mathcal{A}_Z} \frac{ \| [ \tau_t^{\Lambda}(A), B
  ] \| }{ \| A \|},
\end{equation}
and note that $C_B(Z; 0) \leq 2 \|B \| \delta_Y(Z)$, where we defined $\delta_Y (Z) = 1$ if $Y\cap Z \neq \emptyset$ and $\delta_Y (Z) = 0$ if $ Y\cap Z = \emptyset$.
A key observation in our proof will be the fact that the dynamics generated by
\[ H_{\Lambda}^{\text{loc}} \, + \, H_X^{\Phi} \, = \, \sum_{x \in \Lambda} H_x \, + \, \sum_{Z \subset X} \Phi(Z) \] remains local.
More precisely, if we define
\begin{equation}\label{eq:tauloc}
\tau_t^{\rm{loc}}( A) \, = \, e^{it \left(H^{\rm loc}_{\Lambda} + H_X^{\Phi} \right)} \, A \, e^{-it \left(H^{\rm loc}_{\Lambda} + H_X^{\Phi} \right)} \quad \mbox{for all} \quad A \in \mathcal{A}_{\Lambda},
\end{equation}
we have that for every $A \in \cA_X$, $\tau_t^{\rm{loc}} (A) \in \cA_X$ for every $t\in \bR$. This implies, recalling the definition (\ref{def:cbxt}), that
\begin{equation}
C_B(X; t) \, = \, \sup_{A \in \mathcal{A}_X} \frac{ \| [ \tau_t^{\Lambda}
(\tau_{-t}^{\text{loc}} (A)), B ] \| }{ \| A \|}\,.
\label{eq:cbx}\end{equation}

Consider the function (setting $\tau_t( \cdot) = \tau_t^{\Lambda}(\cdot)$)
\begin{equation*}
f(t) \, := \, \left[ \tau_t \left( \tau_{-t}^{\rm{loc}}(A) \right), B \right],
\end{equation*}
for $A \in \mathcal{A}_X$, $B \in \mathcal{A}_Y$, and $t \in
\mathbb{R}$. It is straightforward to verify that
\begin{equation} \label{eq:derf}
f'(t) = i \sum_{Z \in S_{\Lambda}(X)} \left[ \tau_t \left( \Phi(Z)
    \right), f(t) \right] - i  \sum_{Z \in \, S_{\Lambda}(X)} \left[
    \tau_t( \tau_{-t}^{{\rm loc}}(A)), \left[ \tau_t \left( \Phi(Z) \right), B \right] \right] .
\end{equation}
As is discussed in \cite[Appendix A]{NOS1}, the first term in the above differential
equation is norm preserving, and therefore we have the bound
\begin{equation} \label{eq:normpresbd}
\| f(t) \| \, \leq \, \| f(0) \| \, + \, 2 \| A \| \sum_{Z \in S(X)} \int_0^{|t|} \| [ \tau_s(\Phi(Z)), B ] \| ds.
\end{equation}
Recalling definition (\ref{def:cbxt}), the above inequality readily
implies that
\begin{equation} \label{eq:recurbd}
C_B(X,t) \leq C_B(X,0) + 2 \sum_{ Z \in \, S(X)} \| \Phi(Z) \| \int_0^{|t|} C_B(Z, s) ds,
\end{equation}
where we have used (\ref{eq:cbx}).
Iterating this inequality, exactly as is done in \cite{NOS1}, see also \cite{NSP},
yields (\ref{eq:lrbd1}) with (\ref{eq:gat}) and (\ref{eq:defda}).
The inequality (\ref{eq:vel}), stated in the corollary, readily follows.
\end{proof}

In many situations, $\Lambda \subset \bZ^{\nu}$ and the bound (\ref{eq:vel})
can be made slightly more explicit (but less optimal) by choosing
$$
F(x) = (1 +|x|)^{-\nu-1},\quad\mbox{and }  C= 2^{\nu+1}\sum_{x\in \mathbb{Z}^\nu}
\frac{1}{(1+\vert x\vert)^{\nu+1}}\, .
$$
In this case we have
\begin{equation}\label{eq:lrexp}
\| [ \tau_t^{\Lambda}(A), B ] \| \, \leq \, 2^{-(\nu+1)} \, \| A \| \| B \| \, \min[ | \partial_{\Phi}X|, | \partial_{\Phi} Y|] \, e^{- (ad(X,Y) -  2 \| \Phi \|_a C |t|)}. \end{equation}
for all $a >0$, with
\[ \| \Phi \|_a = \sup_{x,y \in \Lambda} e^{a|x-y|} (1 + |x - y|)^{\nu+1} \, \sum_{X \ni x,y} \| \Phi (X) \| < \infty \,. \]
Eq. (\ref{eq:lrexp}) gives the upper bound $2\| \Phi \|_a C / a$ for the speed of propagation in these systems.

One application of the general framework used in Theorem \ref{thm:phi}
concerns systems comprised of finite clusters with possibly unbounded interactions
within each cluster but only bounded interactions between clusters.
For such systems, by adjusting $\Gamma$ and $d(x,y)$, Theorem \ref{thm:phi} still applies.

\setcounter{equation}{0}

\section{Harmonic Lattice Systems} \label{sec:harm}
In this section, we present our second example of Lieb-Robinson bounds for systems
with unbounded Hamiltonians. Let $L$ and $\nu$ be positive integers. We will consider harmonic Hamiltonians defined on cubic subsets $\Lambda_L \, = \, \left( -L, L \right]^{ \nu}  \cap \mathbb{Z}^{\nu}$. Specifically, for $j = 1, \dots , \nu$ and real parameters $\lambda_j \geq 0$ and $\omega > 0$, we will analyze the Hamiltonian
\begin{equation} \label{eq:harham}
H_L^h \, =\, H_L^{h}( \{ \lambda_j \}, \omega) \, = \,  \sum_{ x \in \Lambda_L} p_{ x }^2 \, +\, \omega^2 \, q_{ x}^2 \, + \,
\sum_{j = 1}^{\nu}  \lambda_j \, (q_{ x } - q_{ x + e_j})^2,
\end{equation}
with periodic boundary conditions (in the sense that $q_{x+e_j} := q_{x-(2L-1)e_j}$ if $x \in \Lambda_L$ but $x+ e_j \not\in \Lambda_L$), acting in the Hilbert space
\begin{equation} \label{eq:hspace}
\mathcal{H}_{\Lambda_L} = \bigotimes_{x \in \Lambda_L} L^2(\mathbb{R}, dq_x).
\end{equation}
Here $\{ e_j \}_{j=1}^{\nu}$ are the canonical basis vectors in $\mathbb{Z}^{\nu}$, and
since, in most calculations, the values of $\lambda_j$ and $\omega$ will be fixed, we
will simply write $H_L^h$ for notational convenience. The quantities $p_x$ and $q_x$, which appear in (\ref{eq:harham}) above, are the single site momentum and position operators regarded as operators on the full Hilbert space $\mathcal{H}_{\Lambda_L}$
by setting (we use here units with $\hbar =1$)
\begin{equation} \label{eq:pandq}
p_x = \idty \otimes \cdots \otimes \idty
\otimes -i \frac{d}{dq} \otimes \idty \cdots \otimes \idty \quad
\mbox{ and } \quad q_x = \idty \otimes \cdots \otimes \idty \otimes q \otimes \idty
\cdots \otimes \idty,
\end{equation}
i.e., these operators act non-trivially only in the $x$-th factor of $\mathcal{H}_{\Lambda_L}$. These operators satisfy the canonical commutation relations
\begin{equation} \label{eq:comm}
[p_x, p_y] \, = \, [q_x, q_y] \, = \, 0 \quad \mbox{ and } \quad
[q_x, p_y] \, = \, i \delta_{x,y},
\end{equation}
valid for all $x, y \in \Lambda_L$. The Hamiltonian $H_L^h$ describes a system of coupled harmonic oscillators (with mass $m=1/2$) sitting at all $x \in \Lambda_L$.

\medskip

Let $\cA_{\Lambda_L}$ be the algebra of all bounded observables on $\cH_{\Lambda_L}$.
The time-evolution generated by the Hamiltonian (\ref{eq:harham}) is the one-parameter
group of automorphisms $\{ \tau_t^{h; \Lambda_L} \}_{t \in \mathbb{R}}$ of $\cA_{\Lambda_L}$,
defined by
\begin{equation} \label{eq:evo}
\tau_t^{h; \Lambda_L}(A) = e^{itH_L^h}Ae^{-itH_L^h} .
\end{equation}
As we will regard the length scale $L$ to be fixed, we will suppress the dependence of the dynamics
on $\Lambda_L$ in our notation, by setting $\tau_t^h (.) = \tau_t^{h;\Lambda_L}$.

\medskip

An important class of observables in $\cA_{\Lambda_L}$ are the Weyl operators. For a bounded, complex-valued function $f: \Lambda \to \bC$, we define the Weyl operator $W(f)$ by
\begin{equation}\label{eq:weyl}
W(f) = e^{i \sum_{x \in \Lambda} \left(q_x \text{Re } f_x + p_x \text{Im } f_x \right)}
\end{equation}
Clearly, $W(f)$ is a unitary operator in $\cA_{\Lambda_L}$ such that
\[ W^{-1} (f) = W^* (f) = W (-f) \, . \]
Moreover, using the well-known Baker-Campbell-Hausdorff formula
\begin{equation}\label{eq:BCH}
e^{A+B} = e^A e^B e^{-\frac{1}{2}[A,B]} \quad \text{if} \quad [A,[A,B]]=[B,[A,B]]=0 ,
\end{equation}
and the commutation relations (\ref{eq:comm}), it follows that Weyl operators satisfy the Weyl relations
\begin{equation} \label{eq:weylrel}
W(f) \, W(g) \, = \, W(g) \, W(f) \, e^{- i {\rm Im}[ \langle f, \, g \rangle]} \, = \, W(f+g) \, e^{-\frac{i}{2} {\rm Im}[\langle f, \, g \rangle]}
\end{equation}
for any bounded $f, g: \Lambda \to \bC$, and that they generate shifts of the position and the momentum operator, in the sense that
\begin{equation}\label{eq:shift}
W^* (f) \, q_x \, W (f) = q_x - \text{Im } f_x \qquad \text{and } W^* (f) \, p_x \, W (f) = p_x + \text{Re } f_x  \, .
\end{equation}

\bigskip

The main result of this section is a Lieb-Robinson bound for the harmonic time-evolution of Weyl operators.

\begin{theorem} \label{thm:harlrb} For any finite $X,Y \subset \mathbb{Z}^{\nu}$, for all $L >0$ such that $X, Y \subset \Lambda_{L}$, and for any functions $f$ and $g$ with ${\rm supp}(f) \subset X$ and  ${\rm supp}(g) \subset Y$, the estimate
\begin{equation}\label{eq:lrbharm}
\left\| \left[ \tau_t^h \left( W(f) \right), W(g) \right] \right\| \, \leq \,
C \, \| f \|_{\infty} \|g \|_{\infty} \, \sum_{x\in X, y\in Y} \,
e^{-\mu \left( d(x,y) - c_{\omega,\lambda} \max \left( \frac{2}{\mu} \, , \, e^{(\mu/2)+1}\right) |t| \right)}
\end{equation}
holds for all $\mu >0$. Here
\begin{equation}
\label{eq:dXY} d(x,y) = \sum_{j=1}^{\nu} \min_{\eta_j \in \, \bZ} |x_j-y_j + 2L \eta_j| \,.
\end{equation}
is the distance on the torus. Moreover
\begin{equation} \label{eq:defk}
C = \left(2+ c_{\omega,\lambda} e^{\mu/2} + c^{-1}_{\omega,\lambda} \right)
\end{equation} with $c_{\omega,\lambda} = (\omega^2 + 4 \sum_{j=1}^{\nu} \lambda_j)^{1/2}$.
\end{theorem}

\begin{corollary}\label{cor:harlrb}
Under the same conditions as in Theorem \ref{thm:harlrb}, for any $0<a<1$, one has
\begin{equation}\label{eqcor:lrbharm}
\left\| \left[ \tau_t^h \left( W(f) \right), W(g) \right] \right\| \, \leq \,
\tilde{C} \, \| f \|_{\infty} \|g \|_{\infty} \, \min(|X|, |Y|) \,
e^{-\mu \left( a d(X,Y) - c_{\omega,\lambda} \max \left( \frac{2}{\mu} \, , \, e^{(\mu/2)+1}\right) |t| \right)}
\end{equation}
where
$$
d(X,Y) = \min_{x \in X, y \in Y} d(x,y)
$$
and
$$
\tilde{C} = C \sum_{z\in \mathbb{Z}^\nu} e^{-\mu (1-a)|z|}\, .
$$
\end{corollary}

\begin{remark}
i) As we will discuss in Remark~\ref{rem:nb1} (see also Lemma~\ref{lm:Htx}), both
Theorem~\ref{thm:harlrb} and Corollary~\ref{cor:harlrb} remain valid in the case $\omega =0$.

\medskip

ii) If we make the further assumption that the sets $X$ and $Y$ have a minimal separation distance, then a stronger, ``small time'' version of (\ref{eq:lrbharm}) holds. Specifically, let $\mu >0$ be given, and assume that $X$ and $Y$ have been chosen with $d(X,Y) > 1+c_{\omega, \lambda}e^{(\mu/2)+1}$. Then for any functions $f$ and $g$ with support in $X$ and $Y$, respectively, one has that
\begin{equation}\label{eq:lrbharm2}
\left\| \left[ \tau_t^h \left( W(f) \right), W(g) \right] \right\| \, \leq \,  t^{2d(X,Y)} \, C \, \| f \|_{\infty} \|g \|_{\infty} \, \sum_{x\in X, y\in Y} \,
e^{-\mu \left( d(x,y) - c_{\omega,\lambda} \max \left( \frac{2}{\mu} \, , \, e^{(\mu/2)+1}\right) |t| \right)}.
\end{equation}
This bound follows from factoring the $t^{2|x|}$ out of (\ref{eq:smalltime}), and then completing the argument as before.

\medskip

iii) In most applications of the Lieb-Robinson bound it is important to obtain an estimate on the group velocity,
referred to as the Lieb-Robinson velocity \cite{NS1,HK,NOS1,bravyi2006,eisert2006,hastings2007,NSP}. Note that we can obtain arbitrarily fast exponential decay in space
at the cost of a worse estimate for the Lieb-Robinson velocity:
\begin{equation}\label{eq:vharm}
v_h(\mu)=c_{\omega,\lambda}\max( \frac{2}{\mu}, e^{(\mu/2)+1})\, .
\end{equation}
The optimal Lieb-Robinson velocity in the above estimates is obtained by choosing $\mu=\mu_0$, the solution of
$$
\frac{2}{\mu}=e^{(\mu/2) +1}\, .
$$
Clearly, $1/2<\mu_0<1$. This gives the following bound for the Lieb-Robinson velocity
in the harmonic lattice: $v_h(\mu_0) = 2c_{\omega,\lambda}/\mu_0\leq 4c_{\omega,\lambda}$.
\end{remark}

Theorem \ref{thm:harlrb} follows from Lemma~\ref{lem:weylevo} and Lemma~\ref{lm:htx}, both proven below. In
Lemma \ref{lem:weylevo}, we derive an explicit formula for the time evolution of a Weyl operator.
This allows us to bound the norm on the l.h.s. of (\ref{eq:lrbharm}) by certain Fourier sums which we then estimate in Lemma \ref{lm:htx}.

For bounded functions $f, g: \Lambda_L \to \bC$, we define the convolution $(f*g): \Lambda_L \to \bC$ by
\begin{equation} \label{eq:defconv}
(f * g)_x \, = \, \sum_{y \in \Lambda_L} f_y \, g_{x-y},
\end{equation}
for any $x \in \Lambda_L$ (if $(x-y) \not\in \Lambda_L$, then we define $g_{x-y}$ through the periodic boundary conditions).

\begin{lemma} \label{lem:weylevo} Let $L$ be a positive integer and consider a bounded function $f: \Lambda_L \to \bC$. Then the harmonic evolution of the Weyl operator $W(f)$ is the Weyl operator given by
\begin{equation} \label{eq:evoweyl}
\tau_t^{h} \left( \, W(f) \, \right) \, = \, W \left( f_t \right),
\quad f_t = f* \overline{h^{(L)}_{1,t}} \, + \, \overline{f} * h^{(L)}_{2,t}\, .
\end{equation}
Here the even functions $h_{1,t}$ and $h_{2,t}$ are given by
\begin{equation} \label{eq:defh1t}
h^{(L)}_{1,t}(x) \, = \, \frac{i}{2} {\rm Im} \left[ \frac{1}{| \Lambda_L|} \sum_{k \in \Lambda_L^*} \left( \gamma(k) \, + \, \frac{1}{\gamma(k)} \right) \, e^{ik \cdot x - 2i \gamma(k) t} \, \right] \,
+ \,  {\rm Re} \left[  \frac{1}{| \Lambda_L|} \sum_{k \in \Lambda_L^*} e^{ik \cdot x - 2i \gamma(k) t} \,\right],
\end{equation}
and
\begin{equation} \label{eq:defh2t}
h^{(L)}_{2,t}(x) \, = \, \frac{i}{2} {\rm Im} \left[ \frac{1}{| \Lambda_L|} \sum_{k \in \Lambda_L^*} \left( \gamma(k) \, - \, \frac{1}{\gamma(k)} \right) \, e^{ik \cdot x - 2i \gamma(k) t} \, \right] ,
\end{equation}
where
\[ \Lambda_L^* = \left\{   \frac{x \pi}{L} \, : \, x \in \Lambda_L \, \right\}\]
and
\begin{equation} \label{eq:defgamma}
\gamma(k; \omega, \{ \lambda_j \}) \, = \, \gamma(k) \, = \, \sqrt{ \omega^2 \, + \, 4 \sum_{j=1}^{\nu} \lambda_j \, \sin^2(k_j/2)}.
\end{equation}
\end{lemma}

The proof of Lemma \ref{lem:weylevo} is given in Section \ref{sec:weylevo}.

\begin{lemma}\label{lm:htx}
Suppose that the functions $h_{1,t}^{(L)}, h_{2,t}^{(L)}: \Lambda_L \to \bC$ are defined as in (\ref{eq:defh1t}), (\ref{eq:defh2t}).
Then
\begin{equation*}
|h^{(L)}_{m,t} (x)| \leq \left( 1 + \frac{1}{2} c_{\omega,\lambda}e^{\mu/2} + \frac{1}{2} c_{\omega,\lambda}^{-1} \right) e^{-\mu \left( |x| - c_{\omega,\lambda} \max \left( \frac{2}{\mu} \, , \, e^{(\mu/2)+1}\right) |t| \right)}
\end{equation*}
for $m=1,2$, all $\mu >0$, $t \in \mathbb{R}$, and $x \in \Lambda_L$. Here we defined $c_{\omega,\lambda} = (\omega^2 + 4 \sum_{j=1}^{\nu} \lambda_j )^{1/2}$ and $|x| = \sum_{j=1}^{\nu} |x_i|$. Note that the bounds are uniform in $L$.
\end{lemma}

The proof of Lemma \ref{lm:htx} can be found in Section \ref{sec:htx}.
Using these two lemmas, we can complete the proof of Theorem \ref{thm:harlrb}.

\begin{proof}[Proof of Theorem~\ref{thm:harlrb}]
Let $f$ and $g$ be functions supported in disjoint sets $X$ and $Y$, respectively, with separation distance $d(X,Y)>0$. Let $L>0$ be large enough so that
$X \cup Y \subset \Lambda_L$. With Lemma~\ref{lem:weylevo} and the Weyl relations (\ref{eq:weylrel}), it is clear that
\begin{equation*}
\left[ \, \tau_t^{h} \left( \, W(f) \, \right), \, W(g) \, \right] \, = \,
W \left(  f_t \right) \, W(g) \, \left( \, 1 \, - \, e^{-i {\rm Im} \left[ \langle g,  f_t  \rangle \right]} \, \right).
\end{equation*}
Using the above formula, it follows that
\begin{equation}\label{eq:thmharm1}
\left\|  \, \left[ \, \tau_t^{h} \left( \, W(f) \, \right), \, W(g) \, \right]  \, \right\| \, \leq \, \left| \, {\rm Im} \left[ \langle g, f_t\rangle\right]\right|
\, \leq \, \left| \,  \langle g, f* \overline{h^{(L)}_{1,t}} \, + \, \overline{f}* h^{(L)}_{2,t} \rangle \, \right| .
\end{equation}
Expanding the first term, we find that
\begin{equation}
\langle g, f * \overline{h^{(L)}_{1,t}} \rangle \, = \, \sum_{y \in \Lambda_L} \overline{g_y} \, \left( f*h^{(L)}_{1,t} \right)_y \, = \, \sum_{y \in Y} \sum_{x \in X} \overline{g_y} \, f_x \, \overline{h^{(L)}_{1,t}(y-x)},
\end{equation}
and therefore the bound
\begin{equation}
\begin{split}
\big| \, \langle \, g, \, f & * \overline{h^{(L)}_{1,t}} \, \rangle \, \big|
\\
\leq \, & \| f \|_{\infty} \, \| g \|_{\infty} \, \sum_{x \in X, y \in Y} \left| h^{(L)}_{1,t} (x-y) \right| \\ \leq \, & \left( 1 + \frac{1}{2} c_{\omega,\lambda}e^{\mu/2} + \frac{1}{2} c_{\omega,\lambda}^{-1} \right) \, \| f \|_{\infty} \, \| g \|_{\infty} \sum_{x \in X, y \in Y} e^{-\mu \left( d(x,y) - c_{\omega,\lambda} \max\left( \frac{2}{\mu} \, , \, e^{(\mu/2)+1}\right) |t| \right)} \label{eq:gharmbd}
\end{split}
\end{equation}
follows from Lemma~\ref{lm:htx}. A similar analysis applies to the second term on the r.h.s.
of (\ref{eq:thmharm1}), yielding (\ref{eq:lrbharm}).
\end{proof}

\subsection{Harmonic Evolution of Weyl Operators}
\label{sec:weylevo}

The goal of this section is to prove Lemma \ref{lem:weylevo}. To this end, we diagonalize the harmonic Hamiltonian $H_L^h$ by introducing Fourier space operators. Consider the set (recall that $\Lambda_L = (-L,L]^{\nu} \cap \bZ^{\nu}$)
\begin{equation*}
\Lambda_L^* \, = \, \left\{ \, \frac{x \pi}{L} \, : \, x \in \Lambda_L \, \right\} .
\end{equation*}
Then it is clear that $\Lambda_L^* \subset ( - \pi, \pi]^{\nu}$ and $| \Lambda_L^* | = (2L)^{\nu} = | \Lambda_L|$. For each $k \in \Lambda_L^*$, we introduce the operators,
\begin{equation} \label{eq:Q+Pk}
Q_k \, = \, \frac{1}{ \sqrt{ | \Lambda_L |}} \sum_{x \in \Lambda_L} e^{- i k \cdot x} q_x \quad \mbox{and} \quad
P_k \, = \, \frac{1}{ \sqrt{ | \Lambda_L |}} \sum_{x \in \Lambda_L} e^{- i k \cdot x} p_x \, .
\end{equation}

One may easily calculate that $Q_k^* = Q_{-k}$ (similarly, $P_k^* = P_{-k}$) for all $k \in \Lambda_L^*$.
Here we have adopted the convention that for $k = (k_1, \dots, k_{\nu}) \in \Lambda_L^*$,
$-k$ is defined to be the element of $\Lambda_L^*$ whose components are given by
\begin{equation*}
(-k)_j \, = \, \left\{ \begin{array}{cc} - k_j, & \mbox{if } |k_j| < \pi, \\ \pi, & \mbox{otherwise.} \end{array} \right.
\end{equation*}
This is reasonable as $e^{i \pi x} = e^{-i \pi x}$ for all integers $x$.
These operators satisfy the following commutation relations
\begin{equation}\label{eq:qpkcom}
[Q_k, Q_{k'}] \, = \, [P_k, P_{k'}] \, = \, 0 \quad \mbox{and} \quad  [Q_k, P_{k'}] \, = \, i \, \delta_{k, -k'},
\end{equation}
for any $k, k' \in \Lambda_L^*$. Furthermore, for any $x \in \Lambda_L$,
\begin{equation} \label{eq:q+px}
q_x \, = \, \frac{1}{ \sqrt{ | \Lambda_L |}} \sum_{k \in \Lambda_L^*} e^{i k \cdot x} Q_k \quad \mbox{and} \quad
p_x \, = \, \frac{1}{ \sqrt{ | \Lambda_L|}} \sum_{k \in \Lambda_L^*} e^{i k \cdot x} P_k \,.
\end{equation}

With the above relations, it is easy to check that the harmonic Hamiltonian (\ref{eq:harham}) can be rewritten as
\begin{equation} \label{eq:hamQP}
H_L^h \, = \, \sum_{k \in \Lambda_L^*} P_k P_{-k} \, + \, \gamma^2 (k) Q_k Q_{-k} \, .
\end{equation}
where we introduced the notation
\begin{equation} \label{eq:defgamma2}
\gamma(k) \, = \, \gamma(k; \{ \lambda_j \}, \omega) \, = \,  \sqrt{ \omega^2 \, + \, 4 \sum_{j=1}^{\nu} \lambda_j \, \sin^2(k_j/2) } \, .
\end{equation}
Observe that $\gamma(k)$ is independent of sign changes in any component of $k$.

Since we have assumed that  $\omega >0$, we have that $\gamma(k) \geq \omega >0$, and therefore, we may diagonalize the Hamiltonian by setting
\begin{equation} \label{eq:beqns}
b_k \, = \, \frac{1}{ \sqrt{2 \gamma(k)}} \, P_k - i \sqrt{ \frac{\gamma(k)}{2}} \, Q_k \quad {\rm and} \quad
b_k^* \, = \, \frac{1}{ \sqrt{2 \gamma(k)}} \, P_{-k} + i \sqrt{ \frac{\gamma(k)}{2}} \, Q_{-k} \, .
\end{equation}
In fact, as a result of this definition, we find that for $k,k' \in \Lambda_L^*$
\begin{equation}\label{eq:bkcom}
[b_k, b_{k'}] \, = \, [b_k^*, b_{k'}^*] \, = \, 0 \quad \mbox{and} \quad  [b_k, b_{k'}^*] \, =  \, \delta_{k, k'},
\end{equation}
and moreover, for each $k \in \Lambda_L^*$,
\begin{equation} \label{eq:Q+Pk1}
Q_k \, = \, \frac{i}{ \sqrt{ 2 \gamma(k)}} \left( b_k \, - \, b_{-k}^* \right) \quad \mbox{and} \quad
P_k \, = \, \sqrt{ \frac{\gamma(k)}{2}} \left( b_k \, + \, b_{-k}^* \right).
\end{equation}
Inserting the above into (\ref{eq:hamQP}), we have that
\begin{equation} \label{eq:diagham}
H_L^h \, = \, \sum_{k \in \Lambda_L^*} \, \gamma(k) \, \left( \, 2 \, b_k^*\,  b_k \, + \, 1 \, \right).
\end{equation}
{F}rom this representation of the Hamiltonian $H_L^h$, we obtain immediately the Heisenberg evolution of the operators $b_k$ and $b_k^*$. In fact, from the commutation relations (\ref{eq:bkcom}), it follows that
\begin{equation} \label{eq:dynbk}
\tau_t^{h}(b_k)  \, = \, e^{-2i \gamma(k) t} \, b_k \quad {\rm and} \quad \tau_t^{h}(b_k^*)  \, = \, e^{2i \gamma(k) t} \, b_k^*
\end{equation}
for all $t \in \mathbb{R}$.

To compute the evolution of the operators $p_x$ and $q_x$, for $x \in \Lambda_L$, we express them in terms of $b_k$ and $b_k^*$. We find
\begin{equation} \label{eq:axQkPk}
\begin{split}
q_x \, = & \, \frac{1}{\sqrt{| \Lambda_L|}} \sum_{k \in \Lambda_L^*} e^{ik \cdot x} Q_k =\frac{i}{\sqrt{2 |\Lambda_L|}} \sum_{k \in \Lambda_L^*}  \frac{e^{ik\cdot x}}{\sqrt{\gamma (k)}} \left( b_k - b_{-k}^* \right) \\
p_x \, = & \, \frac{1}{\sqrt{| \Lambda_L|}} \sum_{k \in \Lambda_L^*} e^{ik \cdot x} P_k =\frac{1}{\sqrt{2 |\Lambda_L|}} \sum_{k \in \Lambda_L^*}  \sqrt{\gamma (k)} \; e^{ik \cdot x} \, \left( b_k + b_{-k}^* \right)\, .
\end{split}
\end{equation}
Therefore
\begin{equation*}
\begin{split}
\tau_t^h (q_x) \, = & \, \frac{i}{\sqrt{2 |\Lambda_L|}} \sum_{k \in \Lambda_L^*}  \frac{e^{ik\cdot x}}{\sqrt{\gamma (k)}} \left( e^{-2i\gamma (k) t} \, b_k - e^{2i \gamma (k) t} b_{-k}^* \right) \\
= & \, \frac{i}{\sqrt{2 |\Lambda_L|}} \sum_{k \in \Lambda_L^*}  \frac{1}{\sqrt{\gamma (k)}} \left( e^{ik\cdot x -2 i \gamma (k) t} \, b_k - e^{-ik \cdot x + 2 i \gamma (k) t}  \, b_k^* \right)
\end{split}
\end{equation*}
and
\begin{equation*}
\tau_t^h (p_x) = \frac{1}{\sqrt{2 |\Lambda_L|}} \sum_{k \in \Lambda_L^*}  \sqrt{\gamma (k)}  \; \left( e^{ik\cdot x -2i \gamma (k) t} \, b_k + e^{-ik \cdot x + 2 i \gamma (k) t} \, b_{k}^* \right)\,.
\end{equation*}
{F}rom (\ref{eq:beqns}) and (\ref{eq:q+px}), it follows that
\begin{equation*}
\begin{split}
\tau_t^h (q_x) \, = & \, \frac{i}{2 |\Lambda_L|} \sum_{k \in \Lambda_L^*}  \frac{e^{ik\cdot x -2 i \gamma (k) t}}{\sqrt{\gamma (k)}} \left( \frac{1}{\sqrt{\gamma (k)}} \sum_{y \in \Lambda_L} e^{-ik\cdot y} p_y - i \sqrt{\gamma (k)} \sum_{y \in \Lambda_L} e^{-ik\cdot y} q_y \right) \\
&-\frac{i}{2 |\Lambda_L|} \sum_{k \in \Lambda_L^*}  \frac{e^{-ik\cdot x+2 i \gamma (k)t}}{\sqrt{\gamma (k)}}
\left( \frac{1}{\sqrt{\gamma (k)}} \sum_{y \in \Lambda_L} e^{ik\cdot y} p_y + i \sqrt{\gamma (k)} \sum_{y \in \Lambda_L} e^{ik\cdot y} q_y \right)
\end{split}
\end{equation*}
which implies
\begin{equation*}
\begin{split}
\tau_t^h (q_x) \, = & \sum_{y \in \Lambda_L} q_y \, \text{Re } \frac{1}{|\Lambda_L|} \sum_{k \in \Lambda_L^*} e^{i k \cdot (x-y) - 2 i \gamma (k) t} - \sum_{y \in \Lambda_L} p_y \, \text{Im } \frac{1}{|\Lambda_L|} \sum_{k \in \Lambda_L^*} \frac{1}{\gamma (k)} e^{i k \cdot (x-y) - 2 i \gamma (k) t}\,.
\end{split}
\end{equation*}
Analogously, we find
\begin{equation*}
\begin{split}
\tau_t^h (p_x) \, = & \sum_{y \in \Lambda_L} p_y \, \text{Re } \frac{1}{|\Lambda_L|} \sum_{k \in \Lambda_L^*} e^{i k \cdot (x-y) - 2 i \gamma (k) t} + \sum_{y \in \Lambda_L} q_y \, \text{Im } \frac{1}{|\Lambda_L|} \sum_{k \in \Lambda_L^*} \gamma (k) e^{i k \cdot (x-y) - 2 i \gamma (k) t}\,.
\end{split}
\end{equation*}
It is then easy to check that
\begin{equation*}
\begin{split}
\tau_t^h \left( \sum_{x \in \Lambda_L} q_x \, \text{Re } f_x + p_x \, \text{Im } f_x \right) = \sum_{x \in \Lambda_L} q_x  \text{Re } (f_t)_x + p_x \text{Im } (f_t)_x
\end{split}
\end{equation*}
with \[ f_{t} = f * \overline{h_{1,t}^{(L)}} + \overline{f} * h_{2,t}^{(L)} \] and where $h_{1,t}^{(L)}$ and $h_{2,t}^{(L)}$ are defined as in (\ref{eq:defh1t}), (\ref{eq:defh2t}). This proves (\ref{eq:evoweyl}).

\begin{remark} \label{rem:nb1} If we consider the Hamiltonian (\ref{eq:harham}) with $\omega=0$, then we can easily obtain analogous formulas for the time evolution of Weyl operators. In fact, if $\omega =0$, we can still define operators $P_k,Q_k$ as in (\ref{eq:Q+Pk}) and, for every $k \in \Lambda^*_L \backslash \{ 0 \}$, operators $b_k$ and $b_k^*$ exactly as in (\ref{eq:bkcom}). In terms of these operators, the Hamiltonian (\ref{eq:harham}) can be expressed, in the case $\omega=0$, as
\begin{equation*}
H_L^h \, (\omega=0) = \, P_{0}^2 \, + \, \sum_{k \in \Lambda_L^* \setminus \{ 0 \}} \, \gamma(k) \, \left( \, 2 \, b_k^* \, b_k \, + \, 1 \,\right).
\end{equation*}
Since $P_0$ commutes with $b_k, b^*_k$, for all $k \neq 0$, we obtain (using the commutation relation (\ref{eq:bkcom}) and (\ref{eq:qpkcom})) that
\begin{equation*}
\begin{split}
\tau_t^{h}(b_k) \, &= \, e^{-2i \gamma(k) t} \,  b_k, \qquad \tau_t^{h}(b_k^*) \, = \, e^{2i \gamma(k) t} \, b_k^*, \\ \tau_t^{h} (P_0) \, &=  \, P_0, \qquad \text{and } \quad \; \, \tau_t^{h} (Q_0) = Q_0 + 2 t P_0 \,.
\end{split}
\end{equation*}
{F}rom these formulae, we find that, in the case $\omega=0$,
\begin{equation*}
\tau_t^{h} \left( W(f) \right) \, = \, W \left( \, f * \overline{h^{(L)}_{0,1,t}} \, + \, \overline{f} * h^{(L)}_{0,2,t} \, \right),
\end{equation*}
with
\begin{equation*}
\begin{split}
h^{(L)}_{0,1,t}(x) \, &= \, \frac{(1-it)}{|\Lambda_L|} \, + \, \tilde{h}^{(L)}_{1,t}(x), \\
h^{(L)}_{0,2,t}(x) \, &= \, \frac{it}{|\Lambda_L|} \, + \, \tilde{h}^{(L)}_{2,t}(x).
\end{split}
\end{equation*}
and where
\begin{equation} \label{eq:defh1tnb}
\begin{split}
\tilde{h}^{(L)}_{1,t}(x) \, = \, & \frac{i}{2} {\rm Im} \left[ \frac{1}{| \Lambda_L|} \sum_{k \in \Lambda_L^* \setminus \{ k_0 \}} \left( \gamma(k) \, + \, \frac{1}{\gamma(k)} \right) \, e^{ik \cdot x - 2i \gamma(k) t} \, \right] \,
\\ &+ \,  {\rm Re} \left[  \frac{1}{| \Lambda_L|} \sum_{k \in \Lambda_L^* \setminus \{ k_0 \}} e^{ik \cdot x - 2i \gamma(k) t} \,\right],
\end{split}
\end{equation}
and
\begin{equation} \label{eq:defh2tnb}
\tilde{h}^{(L)}_{2,t}(x) \, = \,\frac{i}{2} {\rm Im} \left[ \frac{1}{| \Lambda_L|} \sum_{k \in \Lambda_L^* \setminus \{ k_0 \}} \left( \gamma(k) \, - \, \frac{1}{\gamma(k)} \right) \, e^{ik \cdot x - 2i \gamma(k) t} \, \right] .
\end{equation}
\end{remark}

\subsection{Estimates on Fourier Sums. Proof of Lemma \ref{lm:htx}}
\label{sec:htx}

The goal of this section is to prove Lemma \ref{lm:htx}. For $x \in \Lambda_L$, let
\begin{equation}\label{eq:Fou}
\begin{split}
H^{(0)}_L (t,x)  &=  {\rm Re} \frac{1}{|\Lambda_L|} \sum_{k \in \Lambda^*_L} e^{i \, k \cdot x \, - \, 2 \, i \, \gamma(k) \, t}  \\
H^{(1)}_L (t,x)  &=  {\rm Im} \frac{1}{|\Lambda_L|} \sum_{k\in \Lambda^*_L} \gamma(k) \; e^{i \, k \cdot x \, - \, 2 \, i \, \gamma(k) \, t} \\
H^{(-1)}_L (t,x) &=  {\rm Im} \frac{1}{|\Lambda_L|} \sum_{k\in \Lambda^*_L} \frac{1}{\gamma(k)} \; e^{i \, k \cdot x \, - \, 2 \, i \, \gamma(k) \, t}\,.
\end{split}
\end{equation}
Since $h^{(L)}_{1,t} (x) = H_L^{(0)} (t,x) + (i/2) (H_L^{(1)} (t,x) + H_L^{(-1)} (t,x))$ and $h^{(L)}_{2,t} (x)= (i/2) ( H_L^{(1)} (t,x) - H_L^{(-1)} (t,x))$, Lemma \ref{lm:htx} follows from the following exponential estimates on $H_L^{(m)} (t,x)$.

\begin{lemma}\label{lm:Htx}
Suppose that $H^{(m)}_L (t,x)$, for $m=-1,1,0$, is defined as in (\ref{eq:Fou}), with $\gamma (k)= (\omega^2 + 4\sum_{j=1}^{\nu} \lambda_j \sin^2 (k_j /2))^{1/2}$, and $\omega \geq 0$. Then we have
\begin{equation}\label{eq:Htx}
\begin{split}
|H^{(0)}_L (t,x)| &\leq \, e^{-\mu \left( |x| - c_{\omega,\lambda} \max \left( \frac{2}{\mu} \, , \, e^{(\mu/2)+1}\right) |t| \right)} \\
|H^{(1)}_L (t,x)| &\leq \, c_{\omega,\lambda} e^{\frac{\mu}{2}} \, e^{-\mu \left( |x| - c_{\omega,\lambda} \max \left( \frac{2}{\mu} \, , \, e^{(\mu/2)+1}\right) |t| \right)} \\
|H^{(-1)}_L (t,x)| &\leq \, c^{-1}_{\omega,\lambda} \, e^{-\mu \left( |x| - c_{\omega,\lambda} \max \left( \frac{2}{\mu} \, , \, e^{(\mu/2)+1}\right) |t| \right)}
\end{split}
\end{equation}
for all $\mu >0$, $x \in \Lambda_L$, $t \in \bR$, and $L >0$.
Here $c_{\omega,\lambda} = (\omega^2 + 4 \sum_{j=1}^{\nu} \lambda_j )^{1/2}$.
\end{lemma}

\begin{proof}[Proof of Lemma \ref{lm:Htx}]
We first prove (\ref{eq:Htx}) for $m=0$.
Since $m=0$ throughout this proof, and also $L$ is fixed, we will use here the shorthand notation $H (t,x)$ for $H^{(0)}_L (t,x)$. We start by expanding the exponent $e^{-2i\gamma(k)t}$;
\begin{equation*}
\begin{split}
H(t,x) = \; & \text{Re} \frac{1}{|\Lambda_L|} \sum_{k \in \Lambda^*_L} e^{ik \cdot x} \sum_{n \geq 0} \frac{(-2it \gamma (k))^n}{n!} \\
= \; & \text{Re} \sum_{n \geq 0}  \frac{(-1)^n 4^n t^{2n}}{(2n)!} \;
\frac{1}{|\Lambda_L|} \sum_{k \in \Lambda^*_L} e^{ik \cdot x} \gamma^{2n} (k) \\ &+ 2 \,
\text{Im} \sum_{n \geq 0}  \frac{(-1)^n 4^n t^{2n+1}}{(2n+1)!} \;
\frac{1}{|\Lambda_L|} \sum_{k \in \Lambda^*_L} e^{ik \cdot x} \gamma^{2n+1} (k)\,.
\end{split}
\end{equation*}
The second term vanishes because $\gamma (-k) = \gamma (k)$. As for the first term we expand the exponent $\gamma^{2n} (k)$. We find
\begin{equation}\label{eq:htx1}
\begin{split}
H(t,x)
= \; & \sum_{n \geq 0}  \frac{(-1)^n 4^n t^{2n}}{(2n)!} \sum_{\stackrel{m_0, m_1, \dots, m_{\nu} \geq 0}{m_0+ \dots + m_{\nu} = n}} \frac{n!}{m_0! m_1! \dots m_{\nu}!}  \, \omega^{2m_0} \\ & \hspace{5cm} \times  \prod_{j=1}^{\nu} (4\lambda_j)^{m_j} \frac{1}{2L} \sum_{\stackrel{k_j = \frac{\pi}{L} \ell:}{ \ell=-L+1, \dots L}}  e^{ik_j x_j}  \sin^{2m_j} (k_j/2)\,.
\end{split}\end{equation}
Next we note that, for $-L < x_j \leq L$,
\begin{equation}\label{eq:mjxj}
\frac{1}{2L} \sum_{\stackrel{k_j = \frac{\pi}{L} \ell:}{ \ell=-L+1, \dots L}}  e^{ik_j x_j}  \sin^{2m_j} (k_j/2) = 0
\end{equation}
if $|x_j| > m_j$. This follows from the orthogonality relation
\[ \frac{1}{2L} \sum_{\stackrel{k = \frac{\pi}{L} \ell:}{ \ell=-L+1, \dots L}} e^{ik x} = \delta_{x,0} \] if $x \in \Lambda_L$, and from the
observation that
\begin{equation}
\begin{split}
e^{ik_j x_j} \sin^{2m_j} (k_j /2) &= e^{ik_j x_j} \frac{(1- \cos k_j)^{m_j}}{2^{m_j}} \\ &= \frac{1}{2^{m_j}} \sum_{\ell = 0}^{m_j} {m_j \choose \ell} \frac{(-1)^{\ell}}{2^{\ell}} \sum_{p=0}^{\ell} {\ell \choose p} e^{i(x_j + 2p-\ell)k_j} \, .\end{split}
\end{equation} Since $-m_j \leq -\ell \leq 2p - \ell \leq \ell \leq m_j$, we obtain (\ref{eq:mjxj}). Since moreover
\[ \Big| \frac{1}{2L} \sum_{\stackrel{k_j = \frac{\pi}{L} \ell:}{ \ell=-L+1, \dots L}}  e^{ik_j x_j}  \sin^{2m_j} (k_j/2) \Big| \leq 1 \] for all $x_j$ and $m_j$, we obtain, from (\ref{eq:htx1}),
\begin{equation}\label{eq:t2=1}
\begin{split}
|H (t,x) | \leq \; & \sum_{n \geq |x|}  \frac{4^n t^{2n}}{(2n)!} \sum_{\stackrel{m_0, m_1, \dots, m_{\nu} \geq 0}{m_0+ \dots + m_{\nu} = n}} \frac{n!}{m_0! m_1! \dots m_{\nu}!}  \, \omega^{2m_0} \prod_{j=1}^{\nu} (4\lambda_j)^{m_j} \\ = \; &
\sum_{n \geq |x|}  \frac{(2 c_{\omega,\lambda} t)^{2n}}{(2n)!}
\end{split}
\end{equation}
where we put $c_{\omega,\lambda} = (\omega^2 + 4 \sum_{j=1}^{\nu} \lambda_j)^{1/2}$.
The previous inequality implies that
\begin{equation} \label{eq:smalltime}
\begin{split}
|H(t,x)| \leq \; & \sum_{n \geq |x|} \frac{(2 c_{\omega,\lambda} |t|)^{2n}}{(2n)!} \leq \frac{(2c_{\omega,\lambda} |t|)^{2|x|}}{(2|x|)!} \, e^{2 c_{\omega,\lambda} |t|}\,.
\end{split}
\end{equation}
Using Stirling formula, we find, for arbitrary $\mu >0$ and for
$|x| > |t|  c_{\omega,\lambda} e^{(\mu/2)+1}$,
\begin{equation*}
\begin{split}
|H(t,x)| \leq \; & e^{-\mu \left( |x| - \frac{2 c_{\omega,\lambda}}{\mu} |t| \right)}\,.
\end{split}\end{equation*}
Since, by definition $|H(t,x)| \leq 1$ for all $x\in \bZ^{\nu}$ and $t \in \bR$, we obtain immediately that
\[ |H(t,x)| \leq e^{-\mu \left( |x| - c_{\omega,\lambda} \text{max} \left( \frac{2}{\mu} \, , \, e^{(\mu/2)+1}\right) |t| \right)} \]
for arbitrary $\mu>0$.

The case $m=1$ is handled analogously.
For the case $m=-1$ we note that
\begin{equation}
H_L^{(-1)}(t,x) = -2 \int_0^t H_L^{(0)}(s,x) ds
\end{equation}
and then use the bound already obtained for the case $m=0$.
\end{proof}

\setcounter{equation}{0}

\section{Lieb-Robinson Inequalities for Anharmonic Lattice Systems}
\label{sec:anharm}

In this section we consider perturbations of the harmonic lattice system described by the Hamiltonian $H_L^{h}$ defined in (\ref{eq:harham}).
Specifically, for a cube $\Lambda_L = (-L,L]^{\nu} \subset \bZ^{\nu}$, we consider the anharmonic Hamiltonian
\begin{equation}\label{eq:anharham}
\begin{split}
H_L &= H_L^h \, + \, \sum_{x\in \Lambda_L} V(q_x) \\ &= \sum_{x \in \Lambda_L} p_x^2 \, + \, \omega^2 \, q_x^2 \, +
\, \sum_{x \in \Lambda_L} \sum_{j = 1}^{\nu}  \lambda_j \, (q_{ x } - q_{ x + e_j})^2 \, +\,  \sum_{x \in \Lambda_L} V (q_x)\,.
\end{split}
\end{equation}
We denote the dynamics generated by $H_L$ on the algebra $\cA_{\Lambda_L}$ by $\tau^L_t$; that is
\[ \tau_t^L (A) = e^{it H_L } A \, e^{-i t H_L } \, \qquad \text{for } A \in \cA_{\Lambda_L}. \]
The main result of this section will provide estimates in terms of the function
\begin{equation*}
F_\mu (r) = \frac{e^{-\mu r}}{(1+r)^{\nu+1}}\, .
\end{equation*}
Since the distance function $d$ is a metric, we clearly have
\begin{equation}\label{eq:Fconv}
\sum_{z\in \Lambda_L} F_\mu(d(x,z))F_\mu(d(z,y))\leq C_\nu F_\mu(d(x,y))
\end{equation}
with
\begin{equation}\label{eq:Cnu}
C_\nu=2^{\nu+1}\sum_{z\in\Lambda_L}\frac{1}{(1+|z|)^{\nu+1}}.
\end{equation}

\begin{theorem}\label{thm:anharm}
Suppose that $V \in C^1 (\bR)$ is real valued with $V' \in L^1 (\bR)$ such
that
\begin{equation}\label{eq:Vassu}
\kappa_{V} = \int \rd w \,  |\widehat{V'} (w)| |w|  < \infty \,.
\end{equation}
Then, for every $\mu\geq 1$, and $\epsilon >0$, there exists a constant $C$, such that
for every pair of finite sets $X,Y \subset \bZ^{\nu}$ and $L >0$ such that $X,Y \subset \Lambda_L$, we have
\begin{equation}\label{eq:anharm}
\Big\| \left[ \, \tau_t^L (W(f)) , W (g) \, \right] \Big\| \leq  C \, \|f\|_{\infty} \| g \|_{\infty} \,
e^{(\mu+\epsilon)v |t|}\sum_{x\in X, y\in Y} \, F_\mu(d(x,y))
\end{equation}
for all bounded functions $f, g$ with $\supp f \subset X$ and $\supp g \subset Y$.
Here
\begin{equation*}
C \, = \,  (2 + c_{\omega, \lambda} e^{\frac{( \mu + \epsilon)}{2}} + c_{\omega, \lambda}^{-1} ) \, \sup_{s \geq 0}  \left[ (1+s)^{\nu + 1} e^{- \epsilon s} \right] ,
\end{equation*}
and
\begin{equation*}
v(\mu + \epsilon) \, = \, v_h(\mu + \epsilon) +\frac{C C_\nu \kappa_V}{\mu+\epsilon}\, ,
\end{equation*}
with $v_h(\mu+\epsilon)$ defined in (\ref{eq:vharm}).
\end{theorem}

\begin{corollary}\label{cor:anharm}
Analogously to Corollary \ref{cor:harlrb}, the theorem implies a bound of the form
\begin{equation*}
\Big\| \left[ \, \tau_t^L (W(f)) , W (g) \, \right] \Big\| \leq   \tilde{C} \, \|f\|_{\infty} \| g \|_{\infty} \,\min (|X|, |Y|) \, e^{- \mu \left( d(X,Y) - (1+\frac{\epsilon}{\mu}) v(\mu+\epsilon) |t| \right)}
\end{equation*}
for all $\mu,\epsilon >0$ and where
$$
\tilde{C}= C \sum_{z\in \mathbb{Z}^\nu} \frac{1}{(1+|z|)^{\nu+1}}\, ,
$$
and $d(X,Y)$ denotes the distance between the supports $X$ and $Y$.
\end{corollary}

\begin{proof}
We are going to interpolate between the time evolution $\tau_t^L$ (generated by the Hamiltonian (\ref{eq:anharham})) and the harmonic time evolution
$\tau_t^{h; \Lambda_L}$ generated by (\ref{eq:harham}); to simplify the notation we will drop all the $L$ dependence in $H_L$ and $H^h_L$ and
in the dynamics $\tau_t^L$ and $\tau_t^{h; \Lambda_L}$. We start by noting that
 \[ \left[ \tau_t \left( W(f) \right), W (g) \right] \, = \, \left[  \tau_s \left( \tau_{t-s}^h \left( W (f) \right) \right), W(g) \right] \Big|_{s=t} \; . \]
This leads us to the study of
\begin{equation}\label{eq:dds1}
\begin{split}
\frac{\rd}{\rd s} \; \Big[ \tau_s \Big(\tau^h_{t-s} \Big( &W (f) \Big)\Big), W(g) \Big] \\ = \; &i \left[ \tau_s \left( \left[ \; \sum_{z \in \Lambda_L} V(q_z),  \tau^h_{t-s} \left( W (f) \right) \right] \right) , W(g) \right] \\ = \; &i \sum_{z \in \Lambda_L} \left[ \, \tau_s \left( \left[ V(q_z), W (f_{t-s}) \right] \right) , W(g) \right]
\end{split}
\end{equation}
where we used Lemma~\ref{lem:weylevo} to compute the harmonic evolution of the Weyl operator $W(f)$, and the shorthand notation
\begin{equation}\label{eq:j}
f_{t} = f * \overline{h}^{(L)}_{1,t} + \overline{f}* h^{(L)}_{2,t}
\end{equation}
to denote the harmonic evolution of the wave function $f$. Using (\ref{eq:shift}), we easily obtain that
\begin{equation*}
\begin{split}
[ V(q_z), W (f_{t-s}) ] = \; &W (f_{t-s}) \left( W^* (f_{t-s}) V (q_z) W(f_{t-s}) - V (q_z) \right) \\ = \; &W(f_{t-s}) \left( V (q_z- \text{Im } f_{t-s} (z)) - V (q_z)\right)\,.
\end{split}
\end{equation*}
Inserting the last equation in (\ref{eq:dds1}) we find
\begin{equation}\label{eq:dds2}
\begin{split}
\frac{\rd}{\rd s} \; \Big[ \tau_s \Big( \tau^h_{t-s} \Big(&W (f)\Big)\Big) , W(g) \Big] \\ = \; &i \sum_{z \in \Lambda_L} \left[ \tau_s \left(W(f_{t-s}) \left( V(q_z - \, \text{Im} f_{t-s} (z)) - V(q_z) \right)  \right) , W(g) \right] \\
= \; &i \sum_{z \in \Lambda_L} \left[ \tau_s \left( \tau^h_{t-s} \left( W (f)\right) \right) , W(g) \right] \, \tau_s \left( V(q_z - \text{Im} f_{t-s} (z)) - V(q_z) \right) \\ &+ i \sum_{z \in \Lambda_L} \tau_s \left( \tau^h_{t-s} \left( W (f) \right) \right) \, \left[  \tau_s \left(  V(q_z - \text{Im} f_{t-s} (z)) - V(q_z) \right) , W(g) \right]\,.
\end{split}
\end{equation}
Next, we define a unitary evolution $\cU (s;\tau)$ by
\[ i\frac{\rd}{\rd s} \cU (s;\tau) = \cL (s) \cU(s;\tau), \quad \text{and }  \cU (\tau;\tau) = 1 \] with the time-dependent generator
\[ \cL (s) = \sum_{z \in \Lambda_L} \tau_s \left( V(q_z - \text{Im} f_{t-s} (z)) - V(q_z) \right) \,. \] (Here $t \geq 0$ is a fixed parameter).
Then, by (\ref{eq:dds2}), we have
\begin{equation*}
\begin{split}
\frac{\rd}{\rd s} \; \Big[ \tau_s &\left( \tau^h_{t-s} \left(W (f)\right) \right), W(g) \Big] \cU (s;0) \\
 & = i \sum_{z \in \Lambda_L} \tau_s \left( W(f_{t-s})\right)  \left[  \tau_s \left( V(q_z -\text{Im} f_{t-s} (z)) - V(q_z)  \right) , W(g) \right] \cU (s;0)
\end{split}
\end{equation*}
which implies that
\begin{equation}\label{eq:intds}
\begin{split}
\Big[ \tau_t \left( W (f) \right)& , W(g) \Big] \cU (t;0)\\ = \; & \Big[ \tau_t^h \left( W (f) \right) , W(g) \Big] \\ &+i \sum_{z \in \Lambda_L} \int_0^t \rd s \; \tau_s \left( W(f_{t-s}) \right) \; \left[  \tau_s \left(  V(q_z - \text{Im} f_{t-s} (z)) - V(q_z) \right) , W(g) \right] \cU (s;0)\,.
\end{split}
\end{equation}
Next, we expand
\begin{equation*}
\begin{split}
\left( V(q_z - \text{Im} f_{t-s} (z)) - V(q_z) \right) = \; &- \, \text{Im} f_{t-s} (z) \int_0^1 \rd r \; V' (q_z - r \, \text{Im} f_{t-s} (z))  \\ = &\; - \, \text{Im} f_{t-s} (z) \int_0^1 \rd r \int \rd w \; \widehat{V'} (w) e^{iw (q_z - \, r \, \text{Im} f_{t-s} (z))}\,.
\end{split}
\end{equation*}
where the Fourier transform $\widehat V'$ is defined as
\[ \widehat{V'} (w) = \int \frac{\rd q}{(2\pi)^{\nu}} \, V' (q) e^{-iq \cdot w}\,. \]
{F}rom (\ref{eq:intds}) we obtain
\begin{equation*}
\begin{split}
\Big[ \tau_t \left(W (f)\right) , W(g) \Big] = \; &\Big[ \tau^h_t \left( W (f) \right) , W(g) \Big] \cU(0;t) \\ &- \, i  \sum_{z \in \Lambda_L} \int_0^t \rd s \; \text{Im} f_{t-s} (z) \int_0^1 \rd r \, \int \rd w \, \widehat{V'} (w) \, e^{-iw \, r \, \text{Im} f_{t-s} (z)} \\ &\hspace{3cm} \times \tau_s \left( W(f_{t-s})\right) \left[ \tau_s \left( e^{i w q_z} \right) , W(g) \right] \cU (s;t)\,.
\end{split}
\end{equation*}
Taking the norm, using the unitarity of $\cU (s;t)$) and assuming $t\geq 0$ for convenience, we obtain
\begin{equation}\label{eq:nor}
\begin{split}
\Big\| \Big[ \tau_t \left(W (f)\right)  , W(g) \Big] \Big\| \leq \; & \Big\| \Big[ \tau_t^h \left(W (f)\right) , W(g) \Big] \Big\| \\ &+  \sum_{z \in \Lambda_L} \int_0^t \rd s \; |\text{Im} f_{t-s} (z)|  \int \rd w |\widehat{V'} (w)| \, \Big\| \left[  \tau_s \left( e^{iw q_z} \right) , W(g) \right] \Big\|\,.
\end{split}
\end{equation}
For any $\epsilon >0$, it is clear from
(\ref{eq:gharmbd}) that we have
\begin{equation*}
\begin{split}
\Big\| \Big[\tau^h_t \left( W (f) \right) , W(g) \Big] \Big\|  & \leq \,  (2 + c_{\omega, \lambda}
e^{\frac{( \mu + \epsilon)}{2}} + c_{\omega, \lambda}^{-1} ) \, \| f \|_{\infty} \, \| g \|_{\infty}
\, e^{(\mu+\epsilon) v_h (\mu+\epsilon)\, t} \, \sum_{x \in X, y \in Y} e^{-(\mu + \epsilon)d(x,y)} \\
& \leq \, C \, \| f \|_{\infty} \, \| g \|_{\infty} \, e^{\tilde{v} \, t } \,
\sum_{x \in X, y \in Y} F_{\mu}(d(x,y)),
\end{split}
\end{equation*}
where we have set $\tilde{v}= (\mu+\epsilon) v_h (\mu+\epsilon)$. Similarly, the bound
\begin{equation}\label{eq:ft-s}
|\text{Im} f_{t-s} (z)| \, \leq \, C \, \| f \|_{\infty} \, e^{\tilde{v} (t-s)}\,
\sum_{x \in X} F_{\mu} (d(z,x)) ,
\end{equation}
follows from an argument as in (\ref{eq:gharmbd}), for all $0 \leq s \leq t$.
Plugging these observations into (\ref{eq:nor}), we find that
\begin{equation*}
\begin{split}
\Big\| \Big[ \tau_t (W (f)) &  , W(g) \Big] \Big\| \\ \leq \; &   C \, \| f \|_{\infty} \, \| g \|_{\infty} \,
e^{\tilde{v} \, t } \, \sum_{x \in X, y \in Y} F_{\mu}(d(x,y)) \\ &+ \,  C \, \| f \|_{\infty} \,
\sum_{z \in \Lambda_L}
\sum_{x \in X} F_{\mu}(d(z,x)) \, \int \rd w \, |\widehat{V'} (w)| \, \int_0^t \rd s \, e^{\tilde{v} (t-s)}  \,
\Big\| \left[  \tau_s \left( e^{iw q_z} \right) , W(g) \right] \Big\|\,.
\end{split}
\end{equation*}
Iterating this inequality $m$ times we obtain
\begin{equation}\label{eq:expa}
\begin{split}
\Big\| \Big[ \tau_t &\left(W (f)\right)  , W(g) \Big] \Big\| \\ \leq
\; & C \, \| f \|_{\infty} \, \| g \|_{\infty} \, e^{\tilde{v} \, t } \, \sum_{x \in X, y \in Y} F_{\mu}(d(x,y)) \\ &+ C \| f \|_{\infty} \| g \|_{\infty} e^{\tilde{v} t} \sum_{x \in X, y \in Y} \sum_{n=1}^{m} \frac{(Ct)^n}{n!} \left( \prod_{j=1}^n \int d w_j |w_j| | \widehat{V'}(w_j)| \right) \\ &\hspace{1.5cm} \times
\sum_{z_1, \dots,z_n \in \Lambda_L} F_{\mu}(d(x,z_1)) \, F_{\mu}(d(z_1,z_2)) \dots F_{\mu}(d(z_n ,y)) \\
&+ \, C^{m+1} \, \| f \|_{\infty} \, \sum_{x \in X}  \,  \left( \prod_{j=1}^m \int d w_j |w_j| | \widehat{V'}(w_j)| \right)  \int_0^t d s_1 \int_0^{s_1} d s_2 \dots \int_0^{s_{m}} d s_{m+1} \\
&\hspace{1.5cm} \times  \sum_{z_1, \dots,z_{m+1} \in \Lambda_L}  F_{\mu}(d(x,z_1)) \, F_{\mu}(d(z_1,z_2|) \dots
F_{\mu}(d(z_m  , z_{m+1}))  \\
&\hspace{3cm} \times  \int \rd w_{m+1} |\widehat{V'} (w_{m+1})| \, e^{\tilde{v}(t-s_{m+1})} \, \Big\| \Big[ \tau_{s_{m+1}} (e^{iw_{m+1} q_{z_{m+1}}} ) , W(g) \Big] \Big\|\,.
\end{split}
\end{equation}
Using (\ref{eq:Fconv}), we find that
\begin{equation*}
\sum_{z_1, \dots, z_n \in \Lambda_L}  F_{\mu} (d(x,z_1)) F_{\mu} (d(z_1 , z_2)) \dots
F_{\mu} (d(z_n,y)) \, \leq \, C_{\nu}^{n} \, F_{\mu}(d(x,y)).
\end{equation*}
As for the error term in (\ref{eq:expa}), we can use the a-priori bound $\|[ \tau_{s_{m+1}} (e^{iw_{m+1} q_{z_{m+1}}} ) , W(g)]\| \leq 2$ to obtain
\begin{equation*}
\begin{split}
2 \, \| f \|_{\infty} \, e^{\tilde{v} \, t} \, \| \widehat{V'} \|_1 \, C \, t & \frac{(C \, \kappa_V \, C_{\nu} \, t )^m}{(m+1)!} \, \sum_{x \in X} \sum_{z_{m+1} \in \Lambda_L} F_{\mu}(d(x,z_{m+1}))   \\
& \leq \,  2 \, \| f \|_{\infty} \, e^{\tilde{v} \, t} \, \| \widehat{V'} \|_1 \, C \, t  \frac{(C \, \kappa_V \, C_{\nu} \, t )^m}{(m+1)!} |X| \sum_{z \in \mathbb{Z}^{\nu}} F_{\mu}(|z|).
\end{split}
\end{equation*}
{F}rom (\ref{eq:expa}), we now conclude that
\begin{equation*}
\begin{split}
\Big\| \Big[ \tau_t \left(W (f)\right)  , W(g) \Big] \Big\|  \leq \; & C\, \| f \|_{\infty} \| g \|_{\infty} \, e^{\tilde{v} t} \, \sum_{x \in X, y \in Y} F_{\mu}(d(x,y))  \, \sum_{n \geq 0} \frac{ (C \, \kappa_V \, C_{\nu} \,t)^n}{n!} \\& +
2 \, \| f \|_{\infty} \, e^{\tilde{v} \, t} \, \| \widehat{V'} \|_1 \, C \, t  \frac{(C \, \kappa_V \, C_{\nu} \, t )^m}{(m+1)!} |X| \sum_{z \in \mathbb{Z}^{\nu}} F_{\mu}(|z|) \\
\leq \; & C\, \| f \|_{\infty} \| g \|_{\infty} \, e^{(\tilde{v} + C \, \kappa_V \, C_{\nu} ) t} \, \sum_{x \in X, y \in Y} F_{\mu}(d(x,y)) \\ & +
2 \, \| f \|_{\infty} \, e^{\tilde{v} \, t} \, \| \widehat{V'} \|_1 \, C \, t  \frac{(C \, \kappa_V \, C_{\nu} \, t )^m}{(m+1)!} |X| \sum_{z \in \mathbb{Z}^{\nu}} F_{\mu}(|z|)  \, .
\end{split}
\end{equation*}
Since this is true for every $m \geq 0$, and since the last term converges to zero as $m \to \infty$, the theorem follows.
\end{proof}

\begin{remark}
Exactly the same proof yields the Lieb-Robinson bounds (\ref{eq:anharm}) for the Hamiltonian \[ \widehat{H}_L = \sum_{x \in \Lambda_L} p_x^2 \, + \, \omega^2 \, q_x^2 \, + \, \sum_{x \in \Lambda_L} \sum_{j = 1}^{\nu}  \lambda_j \, (q_{ x } - q_{ x + e_j})^2 \, +\,  \sum_{x \in \Lambda_L} V (p_x)\,. \] 
Moreover, one can see from the proof that the on-site nature of the anharmonic perturbation does not play an important role here. For example the same technique can be used to establish Lieb-Robinson bounds for the dynamics generated by the Hamiltonian \[ \widetilde{H}_L = \sum_{x \in \Lambda_L} p_x^2 \, + \, \omega^2 \, q_x^2 \, + \,\sum_{x \in \Lambda_L} \sum_{j = 1}^{\nu}  \lambda_j \, (q_{ x } - q_{ x + e_j})^2 \, +\,  \sum_{x \in \Lambda_L} \sum_{j = 1}^{\nu}  \left( V_1 (q_x - q_{x+e_j})  + V_2 (p_x - p_{x+e_j}) \right) \] if both $V_1$ and $V_2$ satisfy the assumption (\ref{eq:Vassu}).
\end{remark}

\setcounter{equation}{0}

\section{Discussion} \label{sec:disc}

\subsection{Other Observables}\label{sec:other_obs}

Theorems \ref{thm:harlrb} and \ref{thm:anharm} give a Lieb-Robinson bound for Weyl operators of
the form
\begin{equation}\label{eq:gen_lrb}
\Vert [\tau_t(W(f)), W(g)]\Vert \leq C \Vert f\Vert_\infty \Vert g\Vert_\infty e^{-\mu(d(X,Y)-v |t|)}
\end{equation}
for $f$ and $g$ supported on finite subsets $X$ and $Y$ of the lattice, where $\tau_t$ is the
dynamics of a harmonic or anharmonic lattice system that satisfies the conditions of these
theorems. From (\ref{eq:gen_lrb}) one can of course immediately obtain a bound for observables
$A$ and $B$ that are finite linear combinations of Weyl operators by a simple application of the
triangle inequality. Two other classes of observables for which we can obtain useful bounds are
worth mentioning.

Note that for every $f: X \to \bC$, $W(f)=e^{ib(f)}$,
with a self-adjoint operator $b(f)$ acting on $\mathcal{H}_X$ (\ref{eq:weyl}), such that $b(sf)=sb(f)$ for every $s \in \bR$. Let $\hat{A}, \hat{B}\in L^1(\mathbb{R})$ be two functions such that $s\hat{A}(s)$ and $s\hat{B}(s)$ are also in $L^1(\mathbb{R})$. Then, it is straightforward to derive a Lieb-Robinson bound for the observables $A(b(f))$ and $B(b(g))$ defined by
\begin{equation}
A(b(f)) = \int ds \hat{A}(s) W(sf), \quad B(b(g)) = \int ds \hat{B}(s) W(sg)\, .
\end{equation}
The result is
\begin{equation}\label{eq:Ahat_lrb}
\Vert [\tau_t(A(b(f))), B(b(g))]\Vert \leq C
\Vert f\Vert_\infty \Vert g\Vert_\infty e^{-\mu (d(X,Y)-v |t|)}\int ds |s\hat{A}(s)| \, \int ds |s\hat{B}(s)|
\end{equation}

By taking derivatives, we can also obtain a Lieb-Robinson bound for the unbounded observables $b(f)$ and $b(g)$
(e.g., $q_x$ and $p_x$). Because $b(f) $ and $b(g)$ are unbounded we apply the Lieb-Robinson bound first on
a common dense domain of analytic vectors (see \cite[Lemma 5.2.12]{bratteli1997}), and find that the
commutator $[\tau_t(b(f)),b(g)]$ has a bounded extension with the following norm bound
\begin{equation}\label{eq:norm_bound}
\Vert [\tau_t(b(f)), b(g)]\Vert \leq C \Vert f\Vert_\infty \Vert g\Vert_\infty e^{-\mu (d(X,Y)-v |t|)}\, .
\end{equation}

\subsection{Exponential Clustering Theorem}\label{sec:exp_clustering}

For a large class of quantum spin systems it was recently proven that a non-vanishing spectral gap
implies exponential decay of spatial correlations in the ground state \cite{NS1,HK,NSP}. Such a result is
often referred to as the Exponential Clustering Theorem. The locality property of the dynamics
provided by a Lieb-Robinson bound is one of the main ingredients in the proof of this result.
In the harmonic case, the clustering properties of the exact ground state can be explicitly
analyzed \cite{cramer2006,schuch2006}, and indeed one finds exponential
decay whenever there is a non-vanishing gap. For the harmonic systems considered here, the gap is non-vanishing iff $\omega>0$. The results of this paper can be used to prove an exponential
clustering theorem for the class of anharmonic lattice systems we consider here. In fact,
following the method of \cite{NSP} (see also \cite{NS1,HK}), the only additional estimate
needed is the following short-time bound.
\begin{lemma}
Let $H_L$ be the Hamiltonian acting on $\Lambda_L = (-L,L]^{\nu} \subset \bZ^{\nu}$ defined in (\ref{eq:anharham}), and $\tau_t^L$ the time-evolution generated by $H_L$. Let $f,g :\Lambda_L \to \bR$ with $\supp f \subset X$, $\supp g \subset Y$, and $X \cap Y = \emptyset$. Then there exists a constant $C= C (\lambda,\omega,\kappa_V) <\infty$ such that
\begin{equation}\label{eq:short_time}
\Vert [\tau_t(W(f)), W(g)]\Vert \leq C \, |t| \, \min (|X|, |Y|) \; \vert \Vert f \Vert_\infty \Vert g\Vert_\infty
\end{equation}
for all $|t| < t_0 (\lambda,\omega,\kappa_V)$.
\end{lemma}
\begin{proof}
Let $H^{(m)}_L (t,x)$, for $m=0,\pm 1$, be the Fourier sums defined in (\ref{eq:Fou}). {F}rom (\ref{eq:smalltime}), we obtain that, for arbitrary $\mu >0$, 
\[ |H^{(0)} (t,x)| \leq (2c_{\omega,\lambda} |t|) \frac{(2c_{\omega, \lambda}|t|)^{2|x|-1}}{(2|x|!)} \leq c_{\omega,\lambda} \, |t| \, e^{(\mu/2)+1} \, e^{-\mu (|x| - \frac{2c_{\omega,\lambda}}{\mu} |t|)}  \] for all $|x| \geq 1$ and $|t| < e^{-(\mu/2)-1} c^{-1}_{\lambda,\omega}$. Since similar estimates hold for $H^{(1)}$ and $H^{(-1)}$ as well, we find, analogously to (\ref{eq:gharmbd}), that, if $\tau_t^h$ denotes the harmonic time-evolution generated by the Hamiltonian (\ref{eq:harham}), 
\begin{equation}\label{eq:lintharLR}
\begin{split}
\left\| \left[ \tau_t^h \left( W(f) \right), W(g) \right] \right\| \, \leq  &\,  C \, |t|  \, \| f \|_{\infty} \|g \|_{\infty} \, \sum_{x\in X, y\in Y} \,
e^{-\mu \left( d(x,y) - \frac{2c_{\omega,\lambda}}{\mu} |t| \right)} \\ \leq  &\,  C \, |t|  \, \| f \|_{\infty} \|g \|_{\infty} \min (|X|, |Y|)
\end{split}
\end{equation}
for all $|t| < e^{-(\mu/2)-1} c^{-1}_{\omega,\lambda}$ (using the assumption that $X \cap Y = \emptyset$), and for a constant $C$ depending only on $\lambda$ and $\omega$.

\medskip

Next we consider the anharmonic time evolution $\tau_t \equiv \tau_t^L$. {F}rom (\ref{eq:nor}), it follows that 
\begin{equation}
\begin{split}
\Big\| \Big[ \tau_t \left(W (f)\right)  , W(g) \Big] \Big\| \leq \; & \Big\| \Big[ \tau_t^h \left(W (f)\right) , W(g) \Big] \Big\| \\ &+  \sum_{z \in \Lambda_L} \int_0^t d s \; |\text{Im} f_{t-s} (z)|  \int d w |\widehat{V'} (w)| \, \Big\| \left[  \tau_s \left( e^{iw q_z} \right) , W(g) \right] \Big\|\,.
\end{split}
\end{equation}
Applying (\ref{eq:lintharLR}) to bound the first term, (\ref{eq:ft-s}) and Corollary \ref{cor:anharm} to bound the second term, we find
\begin{equation}
\begin{split}
\Big\| \Big[ \tau_t \left(W (f)\right)  , W(g) \Big] \Big\| \leq \; & C \, |t|  \, \| f \|_{\infty} \|g \|_{\infty} \min (|X|, |Y|) 
\end{split}
\end{equation}
for a constant $C$ depending only on $\lambda,\omega$ and on the constant $\kappa_V$ defined in (\ref{eq:Vassu}), and for all $|t|$ sufficiently small (depending on $\lambda$, $\omega$, and $\kappa_V$). 
\end{proof} 
 
As a consequence of these considerations one obtains the following theorem.

\begin{theorem}
Let $H$ be the Hamiltonian of a harmonic or anharmonic lattice model satisfying the
conditions of Theorem \ref{thm:harlrb} or \ref{thm:anharm}, and suppose $H$ has a unique
ground state $\Omega$ and a spectral gap $\gamma$ above the ground state. Denote
by $\langle \,\cdot\,\rangle$ the expectation in the state $\Omega$. Then, for any functions
$f$ and $g$ with supports $X$ and $Y$ in the lattice we have the following estimate:
\begin{equation}
\left| \langle W(f) W(g)\rangle - \langle W(f) \rangle \langle W(g)\rangle\right|
\leq C \Vert f\Vert_\infty \, \Vert g\Vert_\infty\,  \Vert \widehat{V^\prime}\Vert_1 \min(|X|,|Y|)
e^{-d(X,Y)/\xi}
\end{equation}
where $\mu\geq 1$ and $\epsilon>0$ are as in Theorem \ref{thm:anharm}
and $\xi$ can be taken to be
\begin{equation}
\xi = \frac{2(\mu+\epsilon)v(\mu+\epsilon) +\gamma}{\mu\gamma}
\end{equation}
and where, if we assume $d(X,Y)\geq \xi$, $C$ is a constant depending only on the lattice.
\end{theorem}

It is straightforward to see that the same bound holds for infinite systems if the
corresponding GNS Hamiltonian has a unique ground state and a spectral gap above it,
and the infinite system is the thermodynamic limit of finite systems that satisfy the conditions
of Theorem \ref{thm:harlrb} or \ref{thm:anharm}.

\section*{Acknowledgements}

This article is based on work supported in part by the U.S. National Science
Foundation under Grant \# DMS-0605342 (B.N. and R.S.). H.R. received support
from NSF Vigre grant \#DMS-0135345. B.S. is on leave from University of Cambridge; his research is supported by a Sofja Kovalevskaja Award of the Humboldt Foundation.

\end{document}